\documentclass[letterpaper,11pt]{article}
\usepackage[usenames,dvipsnames]{xcolor}
\usepackage{amsfonts}
\usepackage{amsmath,amsthm,amssymb,dsfont,pifont}
\usepackage{enumerate}
\usepackage[english]{babel}
\usepackage{graphicx}	
\usepackage{subcaption}
\usepackage[margin=1in]{geometry}
\usepackage{url}
\usepackage{todonotes}
\usepackage{bbm}
\usepackage{caption}
\usepackage{comment}
\usepackage{cite}
\usepackage{physics}
\usepackage{breqn}

\usepackage{thm-restate}
\usepackage{epsfig}
\usetikzlibrary{shapes.symbols,patterns}
\usepackage{hyperref}[breaklinks]
\hypersetup{colorlinks=true,citecolor=blue,linkcolor=blue,filecolor=blue,urlcolor=blue}
\usepackage{cleveref}

\usepackage{nicefrac}
\usepackage{mathtools}

\usepackage{caption}
\usepackage{algorithm}
\usepackage{algpseudocode}
\algnewcommand{\Input}[1]{\Statex \textbf{Input:} #1}
\algnewcommand{\Output}[1]{\Statex \textbf{Output:} #1}

\algrenewcommand{\algorithmicrequire}{\textbf{Input:}}
\algrenewcommand{\algorithmicensure}{\textbf{Output:}}
 
\theoremstyle{plain}
\newtheorem{theorem}{Theorem}
\newtheorem{lemma}[theorem]{Lemma}

\newtheorem{proposition}[theorem]{Proposition}
\newtheorem{conjecture}[theorem]{Conjecture}

\theoremstyle{definition}
\newtheorem{definition}[theorem]{Definition}

\newcommand*{\E}{\mathop{\mathbb{E}}}

\newcommand*{\cC}{\mathcal{C}}

\newcommand*{\cF}{\mathcal{F}}

\newcommand*{\cP}{\mathcal{P}}

\newcommand*{\U}{\mathcal{U}}

\newcommand*{\Z}{\mathbb{Z}}
\newcommand*{\R}{\mathbb{R}}

\newcommand*{\C}{\mathbb{C}}
\newcommand{\F}{\mathbb{F}}

\newcommand*{\eps}{\varepsilon}

\usepackage{preamble}

\def\Fdnn{\F_d^{2n}}
\def\Fdn{\F_d^{n}}
\def\Fd{\F_d}

\usepackage[normalem]{ulem}

\usepackage{enumitem}
\usepackage{authblk}
\usepackage{thm-restate}
\def\hat{\widehat}
\def\hU{\hat{U}}

\let\norms\norm
\def\norm{\norms*}
\let\abss\abs
\def\abs{\abss*}
\def\Pnorm{P}

\begin{document}
\title{On Clifford hierarchy testing and near-extremizers of noncommutative uniformity norms}
 \author{Zongbo (Bob) Bao\thanks{Centrum Wiskunde \& Informatica (CWI) and QuSoft, Amsterdam, \texttt{zongbo.bao@cwi.nl}} \;\;\;Jop Bri{\"e}t \thanks{Centrum Wiskunde \& Informatica (CWI) and   QuSoft, Amsterdam, and Leiden University \texttt{j.briet@cwi.nl}}\;\;\; Davi Castro-Silva \thanks{University of Cambridge, Cambridge, \texttt{dd654@cam.ac.uk}}\\ {Philippe van Dordrecht} \thanks{Centrum Wiskunde \& Informatica (CWI) and QuSoft, Amsterdam, \texttt{Philippe.Dordrecht@cwi.nl}}\;\;\; Jonas Helsen\thanks{Centrum Wiskunde \& Informatica (CWI) and QuSoft, Amsterdam, \texttt{jonas@cwi.nl}}
   }
   \date{\today}
\maketitle
\vspace{-2.5em}
\begin{abstract}
We consider the problem of testing whether an unknown unitary is close to a specified level of the Clifford hierarchy.
Bu, Gu, and Jaffe proposed a candidate tester for this task based on a connection with noncommutative analogues of the Gowers uniformity norms.
The complexity of this tester---whose analysis depends on a robust characterization of the near-extremizers of these norms---was left open.
We establish such a characterization for the fourth noncommutative uniformity norm and, as a consequence, obtain an efficient tester for the third level of the Clifford hierarchy.
We further discuss possible routes toward resolving the problem of testing for all higher levels, highlighting the main barriers that remain.
\end{abstract}

\section{Introduction}
The challenge of manipulating quantum information in the presence of noise led to the development of quantum error-correcting codes. The task of performing gates on data encoded in such error correction codes, while controlling the propagation of errors, lies at the heart of quantum fault tolerance. Stabilizer codes are the most widely used class of quantum error-correcting codes, because of their close connection to classical linear codes and relative ease of manipulation. In these codes the so-called Clifford gates can in many cases be implemented fault-tolerantly.
However, Clifford operations alone are insufficient for universal quantum computation, being classically simulable by the Gottesman--Knill theorem.
Gottesman and Chuang addressed this limitation through gate teleportation, a framework that enables non-Clifford operations using only Clifford gates together with special resource states known as magic states \cite{gottesman1999teleportation}.
In the same work, they introduced the \emph{Clifford hierarchy}, a nested sequence of gate sets that captures the operations implementable in this framework.

Since then, the Clifford hierarchy has been studied extensively, and has found numerous applications in quantum computation and the design of fault-tolerant protocols
\cite{bravyi2013classification, warman2025transversal, hu2025climbing, xu2026controlled, manjunath2026universal}.
Moreover, its algebraic structure has been a topic of substantial study
\cite{anderson2024groups, chen2024characterising, cui2017diagonal, de2021efficient, gottesman1999teleportation, he2024permutation, pllaha2020weyl, rengaswamy2019unifying, zeng2008semi, zhou2000methodology}. However, many fundamental questions about this algebraic structure remain unresolved.

In this paper, we study the problem of \emph{testing membership in the Clifford hierarchy}. We do this for two main reasons.
First, operations from the Clifford hierarchy---particularly those in the first three levels---play a central role in quantum computation and quantum error correction \cite{kobayashi2025clifford}, making it natural to ask whether efficient testing algorithms for them exist.
The second level of the Clifford hierarchy coincides with the Clifford group. The problems of stabilizer and Clifford testing were first considered in~\cite{low2009learning,wang2011property} and amount to determining whether an unknown state or unitary is close to a stabilizer state or Clifford operator, or far from every such object.
These problems have seen significant recent progress
\cite{Gross_2021, grewal2024improved, arunachalam2024polynomialtimetoleranttestingstabilizer, bao2024toleranttestingstabilizerstates, mehraban2025improved, iyer2024tolerant, hinsche2025clifford, arunachalam2025learning, castro2026algorithmic}.
Testing membership in higher levels of the Clifford hierarchy can therefore be viewed as a natural extension of this line of work.

Second, the problem of testing Clifford hierarchy membership is closely connected to inverse theorems for a noncommutative analogue of the Gowers uniformity norms, introduced in \cite{bu2025quantum} under the name \emph{quantum uniformity measure}.
In additive combinatorics, the uniformity norms first appeared in the foundational works of Gowers~\cite{gowers1998new,Gowers2001}.
Deep inverse theorems show that these norms detect polynomial structure in functions over abelian groups \cite{  green2005inverse, TaoZiegler:2010,eisnertao2012,tao2007structure}, which make them useful in the analysis of algebraic property testing algorithms~\cite{Samorodnitskylowdegreetests, hatami2019higher}.
Remarkably, the analysis of stabilizer and Clifford testing algorithms leads to inverse theorems for these norms in the setting of quantum states and unitaries, thereby extending their connection with polynomial structure to the Hilbert space and noncommutative setting of quantum information.
Here we continue this research direction for higher orders of uniformity.

\subsection{Main Results}

Our main result is an efficient tolerant tester for the third level of the Clifford hierarchy $\cC^{(3)}$, which is the first level beyond the Clifford group and contains, for instance, the $T$ gate.
The tester is tolerant in the sense that it can identify not only strict membership in the third level of the Clifford hierarchy, but also approximate membership.
To make this precise, given a unitary $U\in \U(d^n)$, we define the \emph{degree-$k$ Clifford fidelity} by
\begin{equation}\label{eq:fidelity}
    \cF_{\cC^{(k)}}(U):=\max_{V\in \cC^{(k)}} \abs{\inner{V,U}}^2,
\end{equation}
where $\inner{V,U}:= d^{-n}\tr(V^*U)$ is the normalized trace inner product.
The algorimthic problem we consider here is to distinguish whether $U$ has degree-$3$ Clifford fidelity close to $1$ or far from $1$.
The formal statement of our main result, which assumes query access to both~$U$ and its inverse, is given as follows:

\begin{theorem}[Tolerant tester for $\cC^{(3)}$]
\label{thm:main-tester}
    There exists a constant $C>1$ such that the following holds.
    Given $\eps>0$ and query access to a unitary $U\in \U(d^n)$ and its inverse $U^*$, there is an efficient algorithm $\textsc{C3Tester}(n, d, U, U^*, \vep)$ that makes $O\p{1/\vep}$ queries to $U$ and $U^*$, and satisfies the following with probability at least $0.9$:
    \begin{enumerate}
        \item if $\cF_{\cC^{(3)}}(U) \ge 1-\vep$,
        $\textsc{C3Tester}\p{n, d, U, U^*, \vep}$ outputs $1$;
        \item if $\cF_{\cC^{(3)}}(U) \le 1-C\vep$,
        $\textsc{C3Tester}\p{n, d, U, U^*, \vep}$ outputs $0$.
    \end{enumerate}
\end{theorem}

The main ingredient in the analysis of this tester is a classification of the near-extremizers of the fourth \emph{Pauli uniformity norm} (or $P^4$ norm) for unitaries; these norms will be defined in \Cref{sec:prelims}.
The family of Pauli uniformity norms first appeared in the literature under the name of ``quantum uniformity measures'' in \cite{bu2025quantum}, and can be regarded as a noncommutative analogue of the Gowers uniformity norms for vector spaces over finite fields.
Bu, Gu, and Jaffe showed that the $P^k$ norm of a unitary matrix~$U$ upper bounds its degree-$(k-1)$ Clifford fidelity $\cF_{\cC^{k-1}}(U)$, and conjectured a partial converse.
Informally, this converse says that a unitary with large $P^k$ norm must correlate well with some unitary in the $(k-1)$-st level of the Clifford hierarchy.
We refer to such statements as \emph{inverse conjectures} for the Pauli uniformity norms, by analogy with inverse theorems for the classical Gowers norms in additive combinatorics.

Inverse conjectures of this kind typically arise in two regimes, sometimes referred to as the ``$99\%$ setting'' and the ``$1\%$ setting.''
In the $99\%$ setting, one assumes that the norm is close to maximal---heuristically, $\|U\|_{P^k} \geq 0.99$---and seeks to show that $U$ is close to an exact extremizer.
In this regime, the inverse problem amounts to a robust classification of near-extremizers.
In the $1\%$ setting, one assumes only that the norm is non-negligible---say, $\|U\|_{P^k} \geq 0.01$---and seeks to prove that $U$ has non-negligible correlation with a structured matrix, such as an element of level $k-1$ of the Clifford hierarchy.
Although these two regimes are qualitatively different, and neither type of theorem formally implies the other, the $99\%$ inverse problem is generally much more accessible than the corresponding $1\%$ inverse problem.

In \Cref{sec:warmup}, we establish inverse theorems in both regimes for $k=2$, using a straightforward adaptation of the standard proof in the commutative setting;
see, e.g., \cite{green2007montreallecturenotesquadratic}.
The case $k=3$ was recently treated in \cite{hinsche2025clifford}, where inverse theorems of this type were used to obtain efficient tolerant testers for Clifford unitaries, corresponding to the second level of the Clifford hierarchy.
Our main technical contribution is to prove the next case in the near-extremal regime, namely a $99\%$ inverse theorem for $k=4$:

\begin{theorem}[99\% inverse theorem for the $P^4$ norm]
\label{thm:intro-main}
    There exists a constant $L>1$ such that for any prime number $d$, natural number $n\geq 1$ and unitary $U\in \U\p{d^n}$, we have
    $$\max_{V\in \cC^{(3)}} \abs{\inner{V,U}}^2 \ge 1 - L\p{1 - \norm{U}_{\Pnorm^4}}.$$
\end{theorem}

The proof of this theorem is based on an inductive argument of Eisner and Tao \cite{eisnertao2012} in the commutative setting, and points to a possible route toward a $99\%$ inverse theorem for all $k\geq 2$.
The main obstruction to carrying out this strategy in full arises from the non-group structure of the Clifford hierarchy: for every $k\geq 3$, the $k$-th level is not closed under multiplication.
This lack of closure introduces barriers to extending the argument beyond $\cC^{(3)}$.
We discuss these barriers in \Cref{sec:barrier}, together with possible approaches for overcoming them.

\subsection{Open Problems}

We highlight three directions for future work.

\begin{itemize}
    \item \emph{Inverse-free testing.}
    Our tester requires query access both to the unknown unitary $U$ and to its inverse $U^*$.
    Is there an efficient inverse-free tester for the third level of the Clifford hierarchy?
    Such a tester is known for Clifford unitaries; see \cite[Section~4]{hinsche2025clifford}.
    The construction in that case relies on an understanding of the commutant of the Clifford group.
    We expect that removing access to $U^*$ at higher levels will require a better understanding of the commutants of higher levels of the Clifford hierarchy.

    \item \emph{The $1\%$ inverse problem for $P^4$.}
    Is there a $1\%$ inverse theorem for the $P^4$ norm?
    For stabilizer states and Clifford unitaries, polynomial-bound $1\%$ inverse theorems are known, leading to fully tolerant testers with polynomial dependence on the parameters    \cite{arunachalam2024polynomialtimetoleranttestingstabilizer, bao2024toleranttestingstabilizerstates, mehraban2025improved, hinsche2025clifford}.
    In the commutative setting, the best known $1\%$ inverse theorem for the fourth Gowers norm has quasipolynomial bounds \cite{milicevic2024quasipolynomial}.
    Can an analogous result be proved for the fourth Pauli uniformity norm?

    \item \emph{Higher levels.}
    Does a $99\%$ inverse theorem hold for Pauli uniformity norms of order larger than~$5$?
    Such results would immediately lead to efficient testers for higher levels of the Clifford hierarchy.
    In \Cref{sec:barrier}, we discuss the main barriers to extending our argument to higher orders, as well as possible ways to overcome them.
\end{itemize}

\section{Preliminaries}\label{sec:prelims}

In this section, we introduce some notation and state useful technical lemmas concerning Weyl operators, Pauli uniformity norms, and the Clifford hierarchy.

\subsection{Weyl Operators and the Heisenberg Group}

We follow the notation of \cite{Gross_2021, heinrich21thesis}.
Let $d$ be a prime number. 
Denote $D=2d$ if $d=2$, and $D=d$ otherwise.
Let $\omega = e^{2\pi i /d}$ and $\tau = (-1)^de^{i\pi /d}$, so that $\tau^2 = \omega$ and the order of $\tau$ is $D$.
Let $\abs{\cdot}: \Fd\rightarrow \lrs{0,1,\dots,d-1} \subset \Z$ be the natural identification map.

\begin{definition}[Weyl operators]
For $v\in \F_d^n$, define $X^v$ and $Z^v$ to be the unitary operators on $\C^{\F_d^n}$ given by $X^v\ket{u} = \ket{u+v}$ and $Z^v\ket{u} = \omega^{u\cdot v}\ket{u}$.
For $a=(u,v)\in \F_d^{2n}$, define the Weyl operator
\begin{equation*}
    W_a := \tau^{-(|u_1v_1| +\cdots+|u_nv_n|)}Z^uX^v.
\end{equation*}
\end{definition}

We use some basic facts about the multiplicative relations satisfied by the Weyl operators.
The standard symplectic inner product on $\F_d^{2n}$ is given by
\begin{equation*}
    [(u,v),(u',v')] := u\cdot v' - u'\cdot v.
\end{equation*}
The following proposition follows directly from the fact that the translation and phase operators satisfy the commutation relations $X^uZ^v = \omega^{-u\cdot v}Z^vX^u$.

\begin{proposition}\label{fact:weyl-calc}
    There exists a map $\beta:\Fdnn\times \Fdnn\rightarrow \Z_D$ such that, for any $a,b\in \Fdnn$, we have
    \begin{align*}
        W_aW_b &= \tau^{\beta(a,b)}W_{a+b}, \\
        W_aW_b &= \omega^{[a,b]}W_bW_a.
    \end{align*}
\end{proposition}

The map~$\beta$ is given explicitly e.g.\ in~\cite[Section 3.3.1]{heinrich21thesis}, where it is also shown to satisfy the following properties:

\begin{lemma}\label{fact:beta}
    For any $a,b,c\in \Fdnn$, we have
    \begin{enumerate}
        \item (Cocycle condition)
            $\beta(a,b+c)+\beta(b,c) = \beta(a,b) + \beta(a+b,c)$;
        \item $\beta(a,b)-\beta(b,a) =  2[a,b]$.
    \end{enumerate}
    If $d=2$, the term $2[a,b]$ above is interpreted as $2\abs{[a,b]} \mod 4$.
\end{lemma}

The Weyl operators multiplicatively generate the \emph{Pauli group}, given by
\begin{equation*}
    \cP_n(d) = \lrs{\tau^tW_a:\: t\in \Z_D, a\in \Fdnn}.
\end{equation*}
It is sometimes useful to think of this group more abstractly as a representation of the Heisenberg group.
The \emph{Heisenberg group} $H(\F_d^{2n})$ is given by the set $\Z_D \times \Fdnn$ with the group law
\begin{align*}
    (t,a)\bullet (t',a') := (t+t'+\beta\p{a,a'}, a+a').
\end{align*}
The \emph{Weil representation} $\rho: H(\F_d^{2n}) \to \cP_n(d)$ defined by $\rho(t, a) = \tau^tW_a$ is then easily seen to give a faithful unitary representation of the Heisenberg group.

\subsection{Pauli Uniformity Norm and Pauli Polynomials}

Let $L\p{d^n}$ be the set of linear maps from $(\C^d)^{\otimes n}$ to itself, and let $\U\p{d^n}\subseteq L\p{d^n}$ be the set of unitary maps.
We will use the normalized Hilbert-Schmidt inner product: for $U,V\in L\p{d^n}$,
\begin{equation*}
    \langle U,V\rangle :=\frac{1}{d^n}\tr[U^*V].
\end{equation*}
This inner product induces the Frobenius norm on $L\p{d^n}$, given by
\begin{equation*}
    \norm{U}_2:=\sqrt{\langle U,U\rangle}.
\end{equation*}
It is easy to see that this norm is unitarily invariant. Moreover,
for unitary matrices, the Frobenius norm and the trace-inner product satisfy the following simple relation.

\begin{lemma}\label{fact:relation-2normdistance-innerproduct}
    For $U,V\in \U\p{d^n}$, we have that
    $\min_{\theta\in [0,2\pi)}\norm{U-e^{i\theta}V}_2^2
    = 2 - 2\abs{\inner{U,V}}$.
\end{lemma}

\begin{proof} For any $\theta\in[0,2\pi), $
    \begin{equation*}
        \big\|U-e^{i\theta}V\big\|_2^2=2-e^{i\theta}\langle U,V\rangle -\overline{e^{i\theta}\langle U,V\rangle}\geq 2-2|\langle U,V\rangle |.
    \end{equation*}
    Choosing $\theta$ such that $e^{i\theta}\langle U,V\rangle=|\langle U,V\rangle|$ shows the result.
\end{proof}

\subsubsection{Pauli uniformity norms}

Next we define the Pauli uniformity norms and state a number of useful properties, which can be found in \cite[Section~1.6]{bu2025quantum}.
These norms can be thought of as noncommutative generalizations of the Gowers uniformity norms~\cite{Gowers2001}.

\begin{definition}
    For $h\in \Fdnn$ and $U \in L\p{d^n}$, define the \emph{Pauli derivative of $U$ in direction $h$} as
    \begin{equation*}
        \partial_h U := W_h UW_h^*U^*.
    \end{equation*}
\end{definition}

\begin{definition}
    For $U\in L\p{d^n}$, define the \emph{$k$-th Pauli uniformity norm} as 
    \begin{equation*}
        \norm{U}_{\Pnorm^k}:= \left(\E_{h_1,\dots,h_k\in \Fdnn} \f{1}{d^n}\tr\lrb{\partial_{h_k}\dots\partial_{h_1}U}\right)^{\frac{1}{2^k}}
    \end{equation*}
\end{definition}

While unclear from the definition, this notion in fact defines a norm on $L(d^n)$ for all $k\geq 2$.\footnote{This fact was communicated to us by Asgar Jamneshan. The case when $k\in \{2, 3\}$ was previously shown by Bu, Gu, and Jaffe.}
A useful fact about these norms is that
\begin{equation}\label{eq:Qknormisboundedbyone}
    \norm{U}_{\Pnorm^k}\leq 1 \text{ for all }U\in \U\p{d^n},
\end{equation}
since in this case $\partial_{h_1}\dots\partial_{h_k}U$ is unitary for all $h_1,\dots,h_k\in \Fdnn$.
One sees directly from the definition that the $\Pnorm^k$ norms satisfy the following inductive or nesting property:
    \begin{equation}\label{fact:inductive}
        \norm{U}_{\Pnorm^k}^{2^k} = \E_{h\in \Fdnn} \norm{\partial_h U}_{P^{k-1}}^{2^{k-1}}.
    \end{equation}

Finally, it is useful to note the following explicit forms for the $P^1$ and $P^2$ norms:
    for all $U\in L\p{d^n}$, we have
    \begin{equation}
\label{lemma:Q1Q2alternateform}
        \norm{U}_{P^1}^2 = \f{1}{d^{2n}}\abs{ \tr\lrb{U}}^2
        \quad\text{and}\quad
        \norm{U}_{P^2}^4 = \f{1}{d^{2n}}\E_{h\in\Fdnn}\abs{ \tr\lrb{\partial_hU}}^2.
    \end{equation}

\subsubsection{Pauli Polynomials and the Clifford hierarchy}

We now define the extremizers of Inequality \eqref{eq:Qknormisboundedbyone}, which we call \emph{Pauli polynomials}.

\begin{definition}[Pauli Polynomials]
    \label{def:quantum-poly}
    Let $\cP^{(0)}:=\lrs{e^{i\theta}I \mid \theta \in [0,2\pi)}$.
    For $k\geq 1$, the Pauli polynomials of degree at most $k$ are defined by
    \begin{equation*}
        \cP^{(k)}:=\lrs{U\in \U\p{d^n} \mid \partial_hU\in \cP^{(k-1)} \text{ for all }h\in \Fdnn}.
    \end{equation*}
\end{definition}

Note that this definition implies that
\begin{equation}\label{eq:polynomialsextremiseQknorm}
    U\in \cP^{(k-1)}\iff \norm{U}_{\Pnorm^k}=1.
\end{equation}
The set of Pauli polynomials coincides with the \emph{Clifford hierarchy}, which is defined as follows:

\begin{definition}[Clifford hierarchy]\label{def:cliff-hierarchy}
    Denote $\cC^{(1)} = \lrs{\tau^pW_h \mid p\in \Z_{D},h\in \Fdnn}$.
    For $k\geq 2$, the $k$-th level of the Clifford hierarchy $\cC^{(k)}$ is defined as
    \begin{equation*}
        \cC^{(k)} := \lrs{U\in \U\p{d^n}\mid U\cC^{(1)}U^* \subseteq \cC^{(k-1)}}
    \end{equation*}  
\end{definition}

In quantum information theory, the group $\cC^{(1)}$ is known as the Pauli group $\cP_n(d)$, and $\cC^{(2)}$ is known as the Clifford group.
The next result easily follows from the above definitions:

\begin{proposition}
    For $k\ge 2$, $\cP^{(k)}$ is equal to $\cC^{(k)}$.
\end{proposition}

The set $\cP^{(k)}$ is readily shown to be closed under left and right multiplication with Pauli matrices, as well as taking inverses.

\begin{lemma}
\label{lemma:polynomials-closed-pauli-mult}
    Let $a,b\in \Fdnn$. For $k\geq 1$, if $U\in \cP^{(k)}$, then $W_aUW_b\in \cP^{(k)}$ and $U^*\in \cP^{(k)}$.
\end{lemma}

Our main goal in this paper is to establish a relationship between the Clifford fidelity, defined in \Cref{eq:fidelity}, and the $P^k$ norm, in particular in the case where the order-$k$ Pauli uniformity norm is close to $1$ (the 99\% regime).
One direction is not too difficult to prove:

\begin{proposition}[Direct inequality, Theorem 8 of \cite{bu2025quantum}]
\label{prop:directinequality}
    For all $U\in \U\p{d^n}$ and $k\ge 1$,
    \begin{equation*}
        \cF_{\cC^{(k-1)}}\p{U} \le \norm{U}_{\Pnorm^k}.
    \end{equation*}
\end{proposition}

This bound can in fact be improved to $\cF_{\cC^{(k-1)}}\p{U} \le \norm{U}_{\Pnorm^k}^2$ by using the Gowers--Cauchy--Schwarz arguments in \cite[Appendix~A]{JAMNESHAN_THOM_2026}.
The other direction is the subject of a conjecture of Bu, Gu, and Jaffe~\cite[Conjecture~9]{bu2025quantum}.

\begin{conjecture}[Inverse Conjecture for the Pauli uniformity norms]
\label{conj:from-bu25}
For any $\eps> 0$, integer $k\geq 2$ and prime number~$d$, there exists a $\delta = \delta(\eps,k,d)>0$ such that the following holds.
For any positive integer~$n$, if a unitary $U\in \U(d^n)$ satisfies $\norm{U}_{\Pnorm^k} \geq \eps$, then $\cF_{\cC^{(k-1)}}\p{U} \geq \delta$.
\end{conjecture}

In this paper, we prove this conjecture for $k= 4$ when $\eps$ is sufficiently close to~1 (\Cref{thm:intro-main}).

\section{Warm-up: an inverse theorem for $k=2$}\label{sec:warmup}

As a warm-up, we prove the first nontrivial case of \Cref{conj:from-bu25}, corresponding to $k=2$.
This case is significantly simpler than that of $k \geq 3$, and follows from a Fourier-analytic argument that is quite standard in the commutative setting (see e.g. \cite{green2007montreallecturenotesquadratic}).

The Fourier expansion of a matrix in $L(d^n)$ is defined in terms of the Weyl operators, which form an orthonormal basis of $L(d^n)$.

\begin{definition}[Fourier expansion of matrices]
    For any $U\in L\p{d^n}$, we can decompose
    \begin{equation*}
        U = \sum_{a\in \Fdnn} \hU(a) W_a,
    \end{equation*}
    where $\hU(a) := \inner{W_a, U}$.
\end{definition}

By orthonormality of the Weyl basis, we have that $\sum_{a\in \Fdnn} \big|\hU(a)\big|^2 = \inner{U,U}$, which equals 1 if~$U$ is a unitary.

\begin{lemma}[The $P^2$-norm equals the $L^4$-norm of the Fourier transform]
For $U\in L\p{d^n}$,
\begin{equation*}
    \norm{U}_{\Pnorm^2}^4=\sum_{a\in \Fdnn}|\widehat{U}(a)|^4
\end{equation*}
\end{lemma}
\begin{proof}
By \Cref{lemma:Q1Q2alternateform},
\begin{equation*}
    \norm{U}_{\Pnorm^2}^4=\E_{h\in\Fdnn}\left|\frac{1}{d^n}\tr[W_hUW_h^*U^*]\right|^2.
\end{equation*}
Using the Fourier expansion of $U$, this is equal to
\begin{equation*}
    \E_{h\in\Fdnn}\Bigg|\sum_{a_1,a_2\in \Fdnn}\widehat{U}(a_1)\overline{\widehat{U}(a_2)}\frac{1}{d^n}\tr[W_hW_{a_1}W_h^*W_{a_2}^*]\Bigg|^2=\E_{h\in\Fdnn}\Bigg|\sum_{a\in \Fdnn}\abs{\widehat{U}(a)}^2\omega^{-[a,h]}\Bigg|^2,
\end{equation*}
where we have used the commutation relations and orthonormality of the Weyl operators. Expanding the outer square, we obtain
\begin{equation*}
    \sum_{a_1,a_2\in \Fdnn}\abs{\widehat{U}(a_1)}^2\abs{\widehat{U}(a_2)}^2\E_{h\in\Fdnn}\omega^{[a_1-a_2,h]}=\sum_{a\in \Fdnn}\abs{\widehat{U}(a)}^4,
\end{equation*}
where the equality follows because $[\cdot,\cdot]$ is a non-degenerate bilinear form.
\end{proof}

\begin{proposition}
[Inverse theorem for $P^2$]\label{prop:Q2inverse}
    For all $U\in \U\p{d^n}$, we have
    \begin{equation*}
        \cF_{\cC^{(1)}}\p{U} \geq \norm{U}_{\Pnorm^2}^4.
    \end{equation*}
\end{proposition}
\begin{proof}
    \begin{equation*}
        \norm{U}_{\Pnorm^2}^4=\sum_{a\in \Fdnn}\abs{\widehat{U}(a)}^4 \le \cF_{\cC^{(1)}}\p{U} \sum_{a\in \Fdnn}\abs{\widehat{U}(a)}^2 
            = \cF_{\cC^{(1)}}\p{U}.
    \end{equation*}
\end{proof}
Note that this inverse theorem is meaningful in both the 99\% setting and the 1\% setting.
Combining it with \Cref{prop:directinequality}, we see that
\begin{equation*}
    \|U\|_{P^2}^4 \leq \cF_{\cC^{(1)}}(U) \leq \|U\|_{P^2},
\end{equation*}
so $\cF_{\cC^{(1)}}(U)$ is polynomially related to $\|U\|_{P^2}$ for unitaries.

\section{A 99\% inverse theorem for $1\le k\le 4$}
\label{sec:99inverse}

In this section we prove our main technical result, \Cref{thm:intro-main}, which gives a 99\% inverse theorem for the Pauli uniformity norm of order $4$.
We begin by showing that Pauli polynomials are ``well-separated'', a fact that will be needed in our proof.

\begin{lemma}[Separation Lemma]\label{lemma:separation}
Let~$k\geq 1$ and let $U\in \cP^{(k)}$ be such that $\norm{U - I}_2\le \delta < 2^{-k+3/2}$.
Then, $U = e^{i\theta}I$ for some $\theta\in [-2\delta, 2\delta]$.
\end{lemma}

\begin{proof}
    We first use induction on $k$ to show that $U=e^{i\theta} I$ for some $\theta\in \R$.
    When $k=1$, if $U\in \cP^{(1)}$ is not of the form $e^{i\theta} I$ for some $\theta$, then $\inner{U,I}=0$ and $\norm{U-I}_2= \sqrt{2}$.
    This establishes the base case.

    Now let $k>1$ and suppose the claim has already been proven for $k-1$.
    Let $U\in \cP^{(k)}$ satisfy $\norm{U - I}_2 < 2^{-k+3/2}$.
    By unitary invariance of the Frobenius norm, for any $h\in \Fdnn$, we have that
    \begin{equation*}
        \norm{W_hUW_h^*-I}_2 < 2^{-k+3/2}.
    \end{equation*}
    By unitary invariance followed by the triangle inequality, we conclude that
    \begin{equation}\label{eq:separation-close}
        \norm{\partial_hU-I}_2 = \norm{W_hUW_h^*-U}_2 < 2\cdot 2^{-k+3/2} = 2^{-(k-1)+3/2}.
    \end{equation}
    Since $\partial_h U\in \cP^{(k-1)}$, from the induction hypothesis we see that $\partial_hU= e^{i\theta_h}I \in \cP^{(0)}$ for all $h\in \Fdnn$, and thus $U\in \cP^{(1)}$.
    If $U$ were not of the form $e^{i\theta} I$, then, as in the base case, we would have $\norm{U-I}_2 = \sqrt{2} > 2^{-k+3/2}$.
    We conclude that $U = e^{i\theta}I$ for some $\theta\in \R$, and so by induction the claim holds for all $k\ge 1$.

    Finally, we bound the value of the phase $\theta$.
    We may assume that $U = e^{i\theta}I$ for some $\theta \in [\pi, \pi)$, and suppose $\norm{U - I}_2 \leq \delta$.
    Note that $\norm{U - I}_2 = \abs{e^{i\theta}-1}$, and that $\abs{e^{i\theta}-1} \geq |\theta|/2$ whenever $\theta \in [\pi, \pi)$.
    We conclude that $|\theta| \leq 2\delta$, as wished.
\end{proof}

The proof of our main result (restated in a slightly different way below) follows along the lines of the arguments given by Eisner and Tao \cite[Theorem~1.1]{eisnertao2012}, adapted to our setting.

\begin{theorem}\label{thm:main}
    For all $1\le k\le 4$ and prime number $d$, there exist constants $0 < \vep_k \leq 1$ and $C_k\geq 1$ such that the following holds.
    For all $n\geq 1$, all $0<\vep\le \vep_k$ and all matrices $U\in \U\p{d^n}$, if $\norm{U}_{\Pnorm^k}^{2^k} \ge 1-\vep$, 
    then $\cF_{\cC^{(k-1)}}\p{U} \ge 1-C_k\vep$.
\end{theorem}

\begin{proof}
We will use induction on $k$.
\\

\noindent
\textbf{Base case $(k=1)$\,:}
It is easy to see that $\cF_{\cC^{(0)}}\p{U} = \abs{\inner{I,U}}^2 = \abs{\f{1}{d^n}\tr\lrb{U}}^2$. 
By \Cref{lemma:Q1Q2alternateform}, this last expression is precisely $\norm{U}_{P^1}^2$.
The conclusion now follows with $\eps_1 = 1$ and $C_1 = 1$.
\\

\noindent
\textbf{Induction $(2\le k\le 4)$\,:}
By \Cref{fact:inductive}, we start with 
\begin{equation*}
    \norm{U}_{\Pnorm^k}^{2^k} = 
        \E_{a\in \Fdnn} \norm{\partial_a U}_{P^{k-1}}^{2^{k-1}}\ge 1-\vep.
\end{equation*}
Let $A = \lrs{a\in \Fdnn\mid \norm{\partial_a U}_{P^{k-1}}^{2^{k-1}}\ge 1-3\vep} \subseteq \Fdnn$.
Since $\norm{\partial_a U}_{P^{k-1}} \leq 1$, we have
\begin{equation*}
    1\cdot \abs{A} + (1-3\vep) \cdot \p{d^{2n} - \abs{A}} \ge \sum_{a\in \Fdnn} \norm{\partial_a U}_{P^{k-1}}^{2^{k-1}} \ge (1-\vep) d^{2n},
\end{equation*}
from which we conclude that $\abs{A}\ge \f{2}{3}d^{2n}$.
Assume $3\epsilon \leq \epsilon_{k-1}$, and let $\delta := 3C_{k-1}\epsilon$.
Then, by the induction hypothesis, for any $a\in A$ there exists $Q_a\in \cP^{(k-2)}$ such that 
    \begin{equation*}
        \abs{\inner{Q_a, \partial_aU}}^2 \ge 1- \delta.
    \end{equation*}
Here, without loss of generality, we let $Q_0 = I$.
Up to replacing each $Q_a$ by $e^{i\theta_a}Q_a$ for some phase $\theta_a$, by \Cref{fact:relation-2normdistance-innerproduct} we can assume that
$\norm{\partial_aU - Q_a}_2\le \sqrt{2\delta}$.
Since for any $h\in \Fdnn$ we have
$\abs{A} + \abs{h-A} \ge \abs{\Fdnn}$, there exist $a,b\in A$ such that $a+b=h$;
fix one such decomposition for each element $h$.
Because $a,b\in A$, we have
\begin{equation*}
    \norm{\partial_aU - Q_a}_2 
        \le \sqrt{2\delta}, \quad
    \norm{\partial_bU - Q_b}_2
        \le \sqrt{2\delta}.
\end{equation*}
The derivative of $U$ in direction $h=a+b$ can be written as
\begin{equation}\label{eq:h=a+b}
    \partial_hU = W_{a+b} U W_{a+b}^* U^*
        = W_b W_a U W_a^* W_b^* U^* 
        = W_b W_a U W_a^* U^* W_b^* W_b U W_b^* U^*
        = W_b \partial_aU W_b^* \partial_b U,
\end{equation}
and so we conclude from the triangle inequality that
\begin{align*}
    \norm{\partial_h U - W_bQ_aW_b^*Q_b}_2
        \le 2\sqrt{2\delta}.
\end{align*}
Now we define $Q_h:=W_bQ_aW_b^*Q_b$, so that the inequality above can be written as
\begin{equation}\label{eq:approx-with-f-Q}
    \norm{\partial_h U - Q_h}_2
        \le 2\sqrt{2\delta}.
\end{equation}
Note that each $Q_h$ thus defined is an element of $\cP^{(k-2)}$, since $\cP^{(k-2)}$ forms a group for $2\le k\le 4$.

The central part of the proof is to show how the collection $\{Q_h:h\in \F_d^{2n}\}$ can be turned into a unitary representation $\phi$ of the Pauli group $\cP_n(d)$.
Towards this goal, we first use Equation \eqref{eq:h=a+b} to show that $Q_{h+h'}$ can be written as a product of $Q_h,Q_{h'}$ and $W_h,W_{h'}$.
For all $h,h'\in \Fdnn$, we have
\begin{align*}
    \norm{\partial_h U - Q_h}_2
        &\le 2\sqrt{2\delta}, \\
    \norm{\partial_{h'} U - Q_{h'}}_2
        &\le 2\sqrt{2\delta}, \\
    \norm{\partial_{h+h'} U - Q_{h+h'}}_2
        &\le 2\sqrt{2\delta}.
\end{align*}
Combining this with \Cref{eq:h=a+b}, we conclude from the triangle inequality that
\begin{equation} \label{eq:combined}
    \norm{Q_{h+h'} - W_{h'}Q_hW_{h'}^* Q_{h'}}_2
        \le 6\sqrt{2\delta}.
\end{equation}
By the Separation Lemma (\Cref{lemma:separation}), if $6\sqrt{2\delta} < 2^{-(k-2)+3/2}$, then for all $h,h'$ there exists $\theta_{h',h}\in \lrb{-12\sqrt{2\delta}, 12\sqrt{2\delta}}$ such that
\begin{equation*}
    e^{i\theta_{h',h}}Q_{h+h'} = W_{h'}Q_hW_{h'}^* Q_{h'}.
\end{equation*}
Recall that $Q_0=I$,
which implies that $\theta_{0,h}=0$ for all $h\in \Fdnn$.
The equation above is equivalent to
\begin{align*}
    Q_{h+h'}^*W_{h+h'}
        &= e^{i\theta_{h',h}}Q_{h'}^*W_{h'}Q_h^*W_{h'}^*W_{h+h'} \\
        &= e^{i\theta_{h',h}}Q_{h'}^*W_{h'} Q_h^*
             W_{h'}^*\tau^{-\beta\p{h',h}}W_{h'}W_h \\
        &= e^{i\theta_{h',h}}\tau^{-\beta\p{h',h}} Q_{h'}^*W_{h'} \cdot Q_h^*W_h,
\end{align*}
where we use \Cref{fact:weyl-calc} in the second equality.
Define the map $\psi:\cP_n(d)\rightarrow \U\p{d^n}$ by $\psi\p{\tau^pW_h} = \tau^p Q_h^*W_h$.
Then we have
\begin{equation*}
    \psi\p{W_{h'+h}} = e^{i\theta_{h',h}}\tau^{-\beta\p{h',h}} \psi\p{W_{h'}} \psi\p{W_h}.
\end{equation*}
By using different orders of association, for any $h,h',h''\in \Fdnn$ we have that
\begin{align*}
    \psi\p{W_{h+h'+h''}} &= e^{i\theta_{h,h'+h''}}\tau^{-\beta\p{h,h'+h''}} \psi\p{W_{h}} \psi\p{W_{h'+h''}} \\
        &= e^{i\theta_{h+h',h''}}\tau^{-\beta\p{h+h',h''}} \psi\p{W_{h+h'}} \psi\p{W_{h''}}.
\end{align*}
Expanding $\psi(W_{h'+h''})$ and $\psi(W_{h+h'})$, we conclude that
\begin{equation*}
    e^{i\theta_{h,h'+h''}} e^{i\theta_{h',h''}} \tau^{-\beta\p{h,h'+h''}} \tau^{-\beta\p{h',h''}}
    = e^{i\theta_{h+h',h''}} e^{i\theta_{h,h'}} \tau^{-\beta\p{h+h',h''}} \tau^{-\beta\p{h,h'}}.
\end{equation*}
Combining this with the cocycle condition on $\beta$ (\Cref{fact:beta}), we obtain
\begin{equation}\label{eq:cocylethetamod2pi}
\theta_{h,h'+h''}+\theta_{h',h''}=\theta_{h,h'}+\theta_{h+h',h''} \mod 2\pi.
\end{equation}
We choose $\varepsilon$ small enough such that $12\sqrt{2\delta}<\frac{\pi}{2}$, which ensures that for any $h,h'\in \Fdnn$, $\abs{\theta_{h,h'}} <\f{\pi}{2}$. This means that \Cref{eq:cocylethetamod2pi} is still true over the real numbers, namely
\begin{equation*}
\theta_{h,h'+h''}+\theta_{h',h''}=\theta_{h,h'}+\theta_{h+h',h''}.
\end{equation*}
If we average the equation above over $h''$, we obtain
\begin{equation*}
    \theta_{h,h'} = b_h + b_{h'} - b_{h+h'},
\end{equation*}
where $b_h := \E_{h''\in\Fdnn}\theta_{h,h''}$.
Note that $b_0 = 0$ because $\theta_{0,h}=0$ for all $h\in \Fdnn$.
Defining the map $\phi:\cP_n(d)\rightarrow \U\p{d^n}$ by
\begin{equation} \label{eq:def_phi}
    \phi\p{\tau^pW_h} = e^{ib_h}\psi(\tau^pW_h) = e^{ib_h} \tau^p Q_h^* W_h,
\end{equation}
we conclude that
\begin{equation*}
    \phi\p{W_{h'+h}} = \tau^{-\beta\p{h',h}} \phi\p{W_{h'}} \phi\p{W_h}.
\end{equation*}

This map $\phi$ is the representation we wanted.
Indeed, it directly follows from the equation above and \Cref{fact:weyl-calc} that for any $h,h'\in \Fdnn$,
\begin{equation*}
    \phi\p{W_{h'}W_h} = \phi\p{\tau^{\beta\p{h',h}}W_{h+h'}}
        = \tau^{\beta\p{h',h}} \phi\p{W_{h+h'}}
        = \phi\p{W_{h'}}\phi\p{W_h}.
\end{equation*}
Furthermore, for any $p,p'\in \Z_D$, $h,h'\in \Fdnn$, we have that
\begin{equation*}
    \phi\p{\tau^{p'}W_{h'}\tau^pW_h} = \tau^{p'}\tau^{p}\phi\p{W_{h'}W_h}
        = \tau^{p'}\tau^{p}\phi\p{W_{h'}}\phi\p{W_h}
        = \phi\p{\tau^{p'}W_{h'}}\phi\p{\tau^{p}W_h}.
\end{equation*}
We conclude that $\phi\p{PQ} = \phi\p{P}\phi\p{Q}$ for all $P,Q\in \cP_n(d)$, and thus $\phi$ forms a unitary representation of the Pauli group $\cP_n(d)$.
This implies that $\phi$ is a representation of the Heisenberg group.

We next show that $\phi$ is unitarily equivalent to the Weil representation of the Heisenberg group, by computing their characters.
Let $p\in \Z_D$ and $h\in \Fdnn$.
The character of the Weil representation at $(p, h)$ is given by
\begin{equation*}
    \tr\p{\tau^p W_h} = \delta_{h,0} \tau^p d^n.
\end{equation*}
For the character of $\phi$ when $h=0$, since $\phi\p{\tau^pI} = e^{ib_0}Q_0^*\tau^p I=\tau^p I$, we have
\begin{equation*}
    \tr\lrb{\phi\p{\tau^pI}} = \tr\lrb{\tau^pI} = \tau^pd^n.
\end{equation*}
For $h\in \Fdnn\setminus\{0\}$, choose some $h'\in \Fdnn$ such that $[h',h]\neq 0$, and let $h'':=h - h'$.
We have
\begin{align*}
    \tr\lrb{\phi\p{\tau^pW_h}} 
        = \tr\lrb{\phi\p{\tau^pW_{h'+h''}}} 
        &= \tau^p\tau^{-\beta\p{h'',h'}}\tr\lrb{\phi\p{W_{h''}}\phi\p{W_{h'}}} \\
    &= \tau^p\tau^{-\beta\p{h'',h'}}\tr\lrb{\phi\p{W_{h'}}\phi\p{W_{h''}}} \\
    &= \tau^p\tau^{\beta\p{h',h''}-\beta\p{h'',h'}}\tr\lrb{\phi\p{W_h}},
\end{align*}
where the second and fourth equalities follow from \Cref{fact:weyl-calc} and the third equality follows from the cyclic property of the trace function.
From \Cref{fact:beta} we conclude that
\begin{equation*}
    \tr\lrb{\phi\p{\tau^pW_h}} = \tau^{2[h',h'']}\tr\lrb{\phi\p{\tau^pW_h}},
\end{equation*}
and since $[h', h''] = [h', h] \neq 0$,
\begin{equation*}
    \tr\lrb{\phi\p{\tau^pW_h}}  = 0 = \delta_{h,0}\tau^pd^n.
\end{equation*}
It follows that these two unitary representations have the same character, which implies that they are unitarily equivalent
(see Lemma 10.1.4 and Corollary 10.2.24 in \cite{CSTbook2018}).
This means that there exists a unitary $V$ such that, for any $h\in \Fdnn$,
\begin{equation*}
    VW_hV^*=\phi\p{W_h}.
\end{equation*}
By the definition of $\phi$ (\Cref{eq:def_phi}), this implies that 
\begin{equation*}
    Q_h = e^{ib_h}\partial_h V.
\end{equation*}
Combining the above identity with \Cref{eq:approx-with-f-Q} and \Cref{fact:relation-2normdistance-innerproduct}, we conclude that, for all $h\in \Fdnn$, we have
\begin{equation*}
    \abs{\inner{\partial_h U,e^{ib_h}\partial_h V}} =
    \abs{\inner{\partial_h U,Q_h}} \geq 
    1 - \frac{1}{2}\norm{\partial_h U - Q_h}_2^2
    \ge 1 - 4\delta.
\end{equation*}
It follows that $\abs{\inner{\partial_h U,\partial_h V}}^2 \ge 1 - 8 \delta$ for all $h\in \Fdnn$.
Now note that we can rewrite
    \begin{align*}
        \E_{h\in \Fdnn} \abs{\inner{\partial_h U, \partial_h V}}^2 =
            \E_{h\in \Fdnn} \abs{\f{1}{d^n} 
                \tr \lrb{\p{W_hUW_h^*U^*}^*\p{W_hVW_h^*V^*}}}^2  = \norm{U^*V}_{P^2}^4,
    \end{align*}
    where the second equality follows by \Cref{lemma:Q1Q2alternateform}.
By definition, there must be a $P\in \cP^{(1)}$ such that $\cF_{\cC^{(1)}}\p{U^*V} = |\langle P,U^*V\rangle|^2$.
It now follows from the $k\!=\!2$ inverse theorem (\Cref{prop:Q2inverse}) that this $P$ has the property
\begin{equation*}
    \abs{\inner{U,VP^*}}^2 = \cF_{\cC^{(1)}}\p{U^*V} \ge \norm{U^*V}_{P^2}^4 \ge 1 - 8\delta.
\end{equation*}
Since $\partial_hV = e^{-ib_h} Q_h\in \cP^{(k-2)}$ for all $h$, we have that $V\in \cP^{(k-1)}$, and thus by \Cref{lemma:polynomials-closed-pauli-mult} we have that $VP^*$ is in $Q^{(k-1)}$ as well.
We conclude that $\cF_{\cC^{(k-1)}}\p{U}\ge 1 - 8\delta = 1- 24C_{k-1} \eps$, which is our desired inequality with $C_k := 24C_{k-1}$.

To determine the parameters $\epsilon_k$ and $C_k$, we used in our arguments the following assumptions:
\begin{align*}
    \eps &\le \eps_k \le \f{1}{3}\eps_{k-1}, \\
    6\sqrt{2\delta} &= 6\sqrt{6C_{k-1}\eps} < 2^{-(k-2)+3/2}, \\
    12\sqrt{2\delta} &= 12\sqrt{6C_{k-1}\eps} < \frac{\pi}{2}.
\end{align*}
For $2\leq k\leq 4$, these assumptions are all satisfied for the choice of parameters
\begin{equation*}
    \eps_k = \f{1}{24^{k}} \quad \text{and} \quad C_k = 24^{k-1},
\end{equation*}
finishing the proof.
\end{proof}
Note that \Cref{thm:intro-main} immediately follows from this result by taking $L = 16\max\{C_4, 1/\eps_4\}$.

\section{A quantum estimator for the Pauli uniformity norms} \label{sec:tester}

We next introduce an algorithm that, when given query access to a unitary $U$ and its inverse $U^*$, accepts~$U$ with probability $\f{1}{2}\p{1 + \norm{U}_{\Pnorm^k}^{2^k}}$.
This algorithm is a natural generalization of earlier algorithms of Low \cite{low2009learning} and Wang \cite{wang2011property}.

\begin{algorithm}
\caption{\textsc{PNormBias}$(n, d, k, U, U^*)$}
\label{alg:u-norm-estimator}

\begin{algorithmic}[1]
\Input{Natural numbers $n,d,k$, oracle access to a unitary $U\in \U\p{d^n}$ and to $U^*$}
\Output{$0$ or $1$}

\item If $k = 1$, prepare two copies of the maximally entangled state
$\ket{\Phi} = \f{1}{\sqrt{d^n}} \sum_{i\in \Fdn} \ket{i,i}$.
Perform a swap test on $\ket{\Phi}$ and $\p{U\otimes I}\ket{\Phi}$ and return the result.
	
\item Sample $h\leftarrow \Fdnn$ uniformly at random, and prepare oracle access to $\partial_h U$ and $\p{\partial_h U}^*$.

\item Return \textsc{PNormBias}$\p{n, d, k-1, \partial_h U, \p{\partial_h U}^*}$.

\end{algorithmic}
\end{algorithm}

\begin{lemma}\label{lemma:estimator}
    The probability that \textsc{PNormBias}$(n, d, k, U, U^*)$ outputs $1$
        is $\f{1}{2}\p{1 + \norm{U}_{\Pnorm^k}^{2^k}}$.
\end{lemma}

\begin{proof}
    We will use induction on $k$.
    When $k=1$, we have
    \begin{align*}
        \Pr\lrb{\textsc{PNormBias}(n, d, 1, U, U^*)=1}
            &= \f{1}{2}\p{1 + \abs{\bra{\Phi}U\otimes I \ket{\Phi}}^2} \\
            &= \f{1}{2}\p{1 + \abs{\f{1}{d^n}\tr\lrb{U}}^2} \\
            &= \f{1}{2}\p{1 + \norm{U}_{P^1}^{2}},
    \end{align*}
    where the last equality follows from \Cref{lemma:Q1Q2alternateform}.
    For any $k>1$, the probability that \textsc{PNormBias}$(n, d, k, U, U^*)$ outputs $1$ is equal to
    \begin{align*}
        \E_{h\in \Fdnn} \Pr\lrb{\textsc{PNormBias}\p{n, d, k-1, \partial_h U, \p{\partial_h U}^*}=1}
        &= \E_{h\in \Fdnn}\lrb{\f{1}{2}\p{1 + \norm{\partial_h U}_{P^{k-1}}^{2^{k-1}}}} \\
        &= \f{1}{2}\p{1+\norm{U}_{\Pnorm^k}^{2^k}},
    \end{align*}
    where the last equality follows from \Cref{fact:inductive}.
\end{proof}

We are finally ready to present our tester for the third level of the Clifford hierarchy, \Cref{alg:C3Tester} below.
This algorithm invokes \textsc{PNormBias} several times to estimate the Pauli uniformity norm of the input unitary, then decides whether the unitary is close or far from the third Clifford level from the estimated value.

\begin{algorithm}
\caption{\textsc{C3Tester}$(n,d, U, U^*,\vep)$}
\label{alg:C3Tester}

\begin{algorithmic}[1]
\Input{Natural numbers $n,d$, oracle access to $U\in \U\p{d^n}$ and $U^*$, real number $0 < \vep \leq 1$.}	
\Output{$0$ or $1$.}

\item Invoke \textsc{PNormBias}$\p{n,d,4,U,U^*}$ for $O\p{\f{1}{\eps^2}}$ times to estimate $\norm{U}_{P^4}^{16}$, up to additive error $\vep$, with probability $9/10$.
Let $E$ be the estimated value.

\item If $E\le 1- 17\eps$, output $0$; otherwise, output $1$.
\end{algorithmic}
\end{algorithm}

\begin{proof}[Proof of \Cref{thm:main-tester}]
    Let $L$ be the constant guaranteed by \Cref{thm:intro-main}.
    We will show that the conclusion of \Cref{thm:main-tester} holds for $C = \max\{2L, 25\}$.
    Note that, for this value of $C$, we may assume without loss of generality that $\eps \leq 0.04$.

    If $\cF_{\cC^{(3)}}\p{U}\ge 1 - \vep$, by \Cref{prop:directinequality} we have that $\norm{U}_{P^4} \ge 1-\vep$, and thus $\norm{U}_{P^4}^{16} \ge 1-16\vep$.
    On the other hand, if $\cF_{\cC^{(3)}}\p{U}\le 1 - 2L\vep$, by \Cref{thm:intro-main} we have that $\norm{U}_{P^4}\le 1-2\vep$.
    Assuming $\eps \leq 0.04$, we conclude that $\norm{U}_{P^4}^{16} \le (1-2\vep)^{16} \leq 1-18\vep$.
    It is then possible to distinguish between these two cases by estimating $\norm{U}_{P^4}^{16}$ up to an additive error $\vep$.
    By \Cref{lemma:estimator} and the Chernoff-Hoeffding inequality, this can be made by running \Cref{alg:u-norm-estimator} $O(1/\vep^2)$ times.
    One can improve the query complexity to $O(1/\vep)$ by a standard application of amplitude estimation on \Cref{alg:u-norm-estimator}.
\end{proof}

Note that a similar tester can be obtained for any given level $k$ of the Clifford hierarchy, as long as one has a corresponding 99\% inverse theorem for the $(k+1)$-st Pauli norm.

\section{Towards a 99\% inverse theorem for $\| U\|_{\Pnorm^k}$ for $k\ge 5$} \label{sec:barrier}

We expect that a $99\%$ inverse theorem for $\norm{U}_{\Pnorm^k}$ holds for every $k$ by extending the inductive strategy from \Cref{sec:99inverse}.
However, the argument in its current form applies only for $1\leq k\leq 4$.
The purpose of this section is to identify the main obstruction to continuing the induction at higher levels, and to record ideas and conjectures that may be useful in overcoming it.
We also discuss a possible approach to the case of $\norm{U}_{\Pnorm^5}$ based on taking two derivatives, inspired by the proof of the inverse theorem for the fifth Gowers uniformity norm in the commutative setting \cite{gowers2017quantitative}.

\subsection{Group barrier and ideas for extending the inductive argument}

To obtain a 99\% inverse theorem for $\norm{U}_{\Pnorm^k}$ for arbitrary $k$ via the inductive argument above, we would need $\cP^{(k-2)}$ to be closed under multiplication.
Indeed, after applying the induction hypothesis in the proof of \Cref{thm:main}, we conclude that for any $h\in A$ there is some $Q_h \in \cP^{(k-2)}$ that approximates the derivative $\partial_h U$.
For $h\notin A$, we instead find $a,b\in A$ satisfying $a+b=h$ and assign $Q_h = W_bQ_aW_b^*Q_b$ to approximate $\partial_h U$
(See \Cref{eq:approx-with-f-Q}).
To make sure that $Q_h$ is also inside $\cP^{(k-2)}$, one seems to require a multiplicative closure property for $\cP^{(k-2)}$.

Note that this property holds for $k\in \{2, 3, 4\}$:
$\cP^{(0)}$ is the circle group times the identity, $\cP^{(1)}$ is the circle group times the Pauli group, and $\cP^{(2)}$ is the Clifford group.
However, for $k=5$, $\cP^{(k-2)}=\cC^{(3)}$ is the third level of Clifford hierarchy, which is not a group;
for instance, the matrix $HT$ (on one qubit) lies in $\cC^{(3)}$, but $(HT)^2$ is not in any level of Clifford hierarchy
\cite[Page 2]{anderson2024groups}.
Therefore, closure under multiplication does not hold for $\cP^{(k-2)}$ for any $k\ge 5$.

What we obtain in the case $k\geq 5$ is that, for any $h\notin A$, $\partial_h U$ is approximated by an element of $\cP^{(k-2),2} := \cP^{(k-2)}\cdot \cP^{(k-2)}$.
There are two obstacles that must be overcome in this case.
First and foremost, after we obtain the final unitary $V$ from the representation-theoretic argument, we only know that $\partial_h V\in \cP^{(k-2)}$ for all $h\in A$, which is not sufficient to imply $V\in \cP^{(k-1)}$.
Following the proof of \Cref{thm:main}, for any $h\in \Fdnn$ there exist $a,b\in A$ such that $a+b=h$, and thus $VW_hV^* = VW_aV^*\cdot VW_bV^*\in \cP^{(k-2),2}$.
Trying to force this unitary to be inside $\cP^{(k-2)}$ leads us to the following conjecture, which was made for an unrelated reason in \cite{de2021efficient}.
    
\begin{conjecture}[Conjecture 1 of \cite{de2021efficient}, rephrased]
    For any $n$-qudit unitary $V$ and natural number $k$, 
        if $VX_iV^*, VZ_iV^*\in \cC^{(k)}$ for all $i=1,2,\dots,n$,
        then we have $V\in \cC^{(k+1)}$,
    where $X_i = I^{\otimes {i-1}}\otimes X\otimes I^{\otimes n-i}$
        and $Z_i = I^{\otimes {i-1}}\otimes Z\otimes I^{\otimes n-i}$.
\end{conjecture}

Note that, although we have a counter-example for the closure of $\cP^{(k-2)}$ under multiplication, in the conjecture above we only consider elements of $\cP^{(k+2)}$ with finite orders and satisfying explicit commutator relations.
The hope is that one might exploit these constraints to derive the desired closure property.

The second obstacle is that we need a separation lemma (i.e. Lemma \ref{lemma:separation}) for $\cP^{(k-2),6}$ (due to \Cref{eq:combined}).
In other words, we need to argue that if a product of six unitaries in $\cP^{(k-2)}\cup \p{\cP^{(k-2)}}^*$ is close to the identity, then it is equal to the identity up to some global phase. 

\subsection{Double derivative barrier for a $P^5$-inverse theorem}

To prove an inverse theorem for $P^5$, it is natural to try to take two derivatives.
Indeed, taking two derivatives of an element of the fourth level of the Clifford hierarchy yields an element of the Clifford group.
A related double-derivative strategy appears in \cite{gowers2017quantitative}, where Gowers and Milićević prove an inverse theorem for the fourth Gowers norm in the commutative setting.

Suppose that $U$ is a unitary with large $P^5$ norm.
Following the first step in the proof of \Cref{thm:main}, one can show that there exists a set $A \subseteq \Fdnn \times \Fdnn$ of size $\abs{A}\geq 0.99 d^{4n}$ such that, for every $(a,b)\in A$, the norm $\norm{\partial_a\partial_b U}_{\Pnorm^3}$ is large.
By the inverse theorem for the $\Pnorm^3$ norm, it follows that for each $(a,b)\in A$ there exists a Clifford unitary $Q_{a,b}\in \cC^{(2)}$ such that $Q_{a,b} \approx \partial_a\partial_b U$.
The next step would be to extend this conclusion from the dense set $A$ to all pairs $(a,b)\in \Fdnn\times \Fdnn$.

Since the double derivatives $\partial_a\partial_b U$ for $(a,b)\in A$ are approximated by Clifford unitaries, one might hope to multiply these approximations together in order to obtain approximations for other derivatives. This works in the first coordinate. For instance, if $(a,b),(a',b)\in A$, then
\begin{equation*}
    \partial_{a+a'}\partial_b U =
    W_{a'} \partial_a\partial_b U W_{a'}^* \partial_{a'}\partial_b U \approx
    W_{a'} Q_{a,b} W_{a'}^* Q_{a',b}.
\end{equation*}
Since the Clifford group is closed under multiplication and conjugation by Pauli operators, the right-hand side lies in $\cC^{(2)}$. Thus one can define
\begin{equation*}
    Q_{a+a',b} := W_{a'} Q_{a,b} W_{a'}^* Q_{a',b} \in \cC^{(2)},
\end{equation*}
as a Clifford approximation to $\partial_{a+a'}\partial_b U$.

The difficulty is that the analogous argument does not seem to work in the second coordinate.
To repeat the same strategy, one would need to express $\partial_a\partial_{b+b'}U$ in terms of $\partial_a\partial_b U$ and $\partial_a\partial_{b'}U$.
However, the relevant expressions do not combine in the required way.
Indeed, for $(a,b),(a,b')\in A$, we have
\begin{align*}
    \partial_a\partial_{b+b'}U &= W_aW_bW_{b'}UW_{b'}^*W_b^*U^*W_a^*UW_bW_{b'}U^*W_{b'}^*W_b^*, \\
    \partial_a\partial_b U &= W_aW_bUW_b^*U^*W_a^*UW_bU^*W_b^*, \\
    \partial_a\partial_{b'} U &= W_aW_{b'}UW_{b'}^*U^*W_a^*UW_{b'}U^*W_{b'}^*.
\end{align*}
There seems to be no direct multiplicative formula analogous to the first coordinate case.

\bibliographystyle{alphaurl}
\bibliography{bib}
\end{document}